\documentclass[10pt]{article}
\usepackage{newtxtext}
\DeclareSymbolFont{timesdigits}{T1}{ptm}{m}{n}
\DeclareMathSymbol{0}{\mathalpha}{timesdigits}{`0}
\DeclareMathSymbol{1}{\mathalpha}{timesdigits}{`1}
\DeclareMathSymbol{2}{\mathalpha}{timesdigits}{`2}
\DeclareMathSymbol{3}{\mathalpha}{timesdigits}{`3}
\DeclareMathSymbol{4}{\mathalpha}{timesdigits}{`4}
\DeclareMathSymbol{5}{\mathalpha}{timesdigits}{`5}
\DeclareMathSymbol{6}{\mathalpha}{timesdigits}{`6}
\DeclareMathSymbol{7}{\mathalpha}{timesdigits}{`7}
\DeclareMathSymbol{8}{\mathalpha}{timesdigits}{`8}
\DeclareMathSymbol{9}{\mathalpha}{timesdigits}{`9}
\usepackage{geometry}
\usepackage{amsmath,amssymb,amsthm}
\usepackage{bm}
\usepackage{xcolor}
\definecolor{graycolor}{gray}{0.9} 
\usepackage{microtype}
\usepackage{setspace} 
\usepackage[utf8]{inputenc}
\usepackage[english]{babel}
\usepackage{times}
\usepackage{array}
\usepackage{soul}
\usepackage{cite}
\usepackage[numbers,sort&compress]{natbib}
\usepackage{natbib}

\usepackage{titlesec} 
\titleformat {\section} [block] {\raggedright \fontsize{10}{10}\selectfont\bfseries} {\thesection. \space} {0pt} {}
\titlespacing {\section} {0pt} {12pt} {6pt}
\titleformat {\subsection} [block] {\raggedright \fontsize{10}{10}\selectfont\itshape} {\thesubsection .\space} {0pt} {}
\titlespacing {\subsection} {0pt} {12pt} {6pt}
\titleformat {\subsubsection} [block] {\raggedright \fontsize{10}{10}\selectfont} {\thesubsubsection .\space} {0pt} {}
\titlespacing {\subsubsection} {0pt} {12pt} {6pt}
\titleformat {\paragraph} [block] {\raggedright \fontsize{10}{10}\selectfont} {} {0pt} {}
\titlespacing {\paragraph} {0pt} {12pt} {6pt}

\usepackage{array} \newcommand{\PreserveBackslash}[1]{\let\temp=\\#1\let\\=\temp}
\newcolumntype{C}[1]{>{\PreserveBackslash\centering}m{#1}}
\newcolumntype{R}[1]{>{\PreserveBackslash\raggedleft}m{#1}}
\newcolumntype{L}[1]{>{\PreserveBackslash\raggedright}m{#1}}
\usepackage{tabularx}
\usepackage{colortbl}
\usepackage{graphicx}
\usepackage{float}
\usepackage[export]{adjustbox}
\usepackage{caption}
\usepackage{fancyhdr} 
\usepackage{lastpage}
\usepackage{layout}
\usepackage{setspace} 
\usepackage{enumitem}
\usepackage{booktabs}
\usepackage{arydshln}
\usepackage{multirow}
\usepackage{color}
\usepackage{hyperref} 
\hypersetup{
	colorlinks=true,
	linkcolor=blue,
	filecolor=blue,
	urlcolor=black,
	citecolor=cyan,
}

\usepackage{url}
\setcitestyle{open={[},close={]},citesep={,\!},numbers}

\newcommand{\dd}{\mathrm{d}}

\newcommand{\CG}{\mathrm{CG}}

\newcommand{\rh}{r_{\mathrm h}}
\newcommand{\rb}{r_{\mathrm b}}
\newcommand{\rc}{r_{\mathrm c}}
\newcommand{\rn}{r_{\mathrm n}}
\newcommand{\rzero}{r_{0}}
\newcommand{\kN}{\kappa_{\mathrm N}}
\newcommand{\TW}{\alpha_{\mathrm W}}

\RequirePackage{amsthm}
\newtheoremstyle{scilight}
{12pt}
{12pt}
{\itshape}
{}
{\bfseries}
{.}
{.5em}
{}
\renewenvironment{proof}{
\normalfont 
\noindent 
\textbf{Proof.} 
\hspace{-0.30em}
\vspace{0pt} 
}{
\vspace{12pt} 
\qed 
}
\theoremstyle{scilight}

\newtheorem{proposition}{Proposition}

\fancypagestyle{firstpage}{
    \setlength{\headsep}{2.2cm}
    
    \setlength{\footskip}{1.5cm}
    \fancyhf{}
    \lhead{\begin{table}[H]
        \centering
        \begin{tabular}{L{2.5cm}C{10cm}C{3.1cm}R{2cm}}
            \includegraphics[scale=0.035]{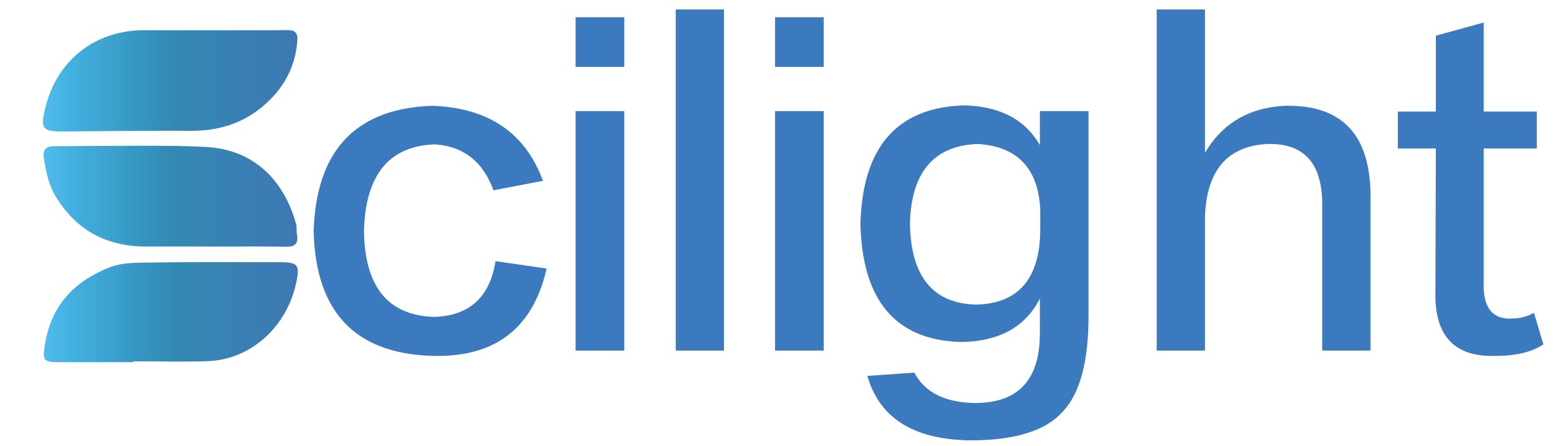} \vspace{-6pt}& \cellcolor{graycolor}\begin{tabular}[c]{@{}c@{}}\textit{International Journal of Gravitation and Theoretical Physics}\\ \href{https://www.sciltp.com/journals/ijgtp}{https://www.sciltp.com/journals/ijgtp}\end{tabular} & \includegraphics[scale=0.014]{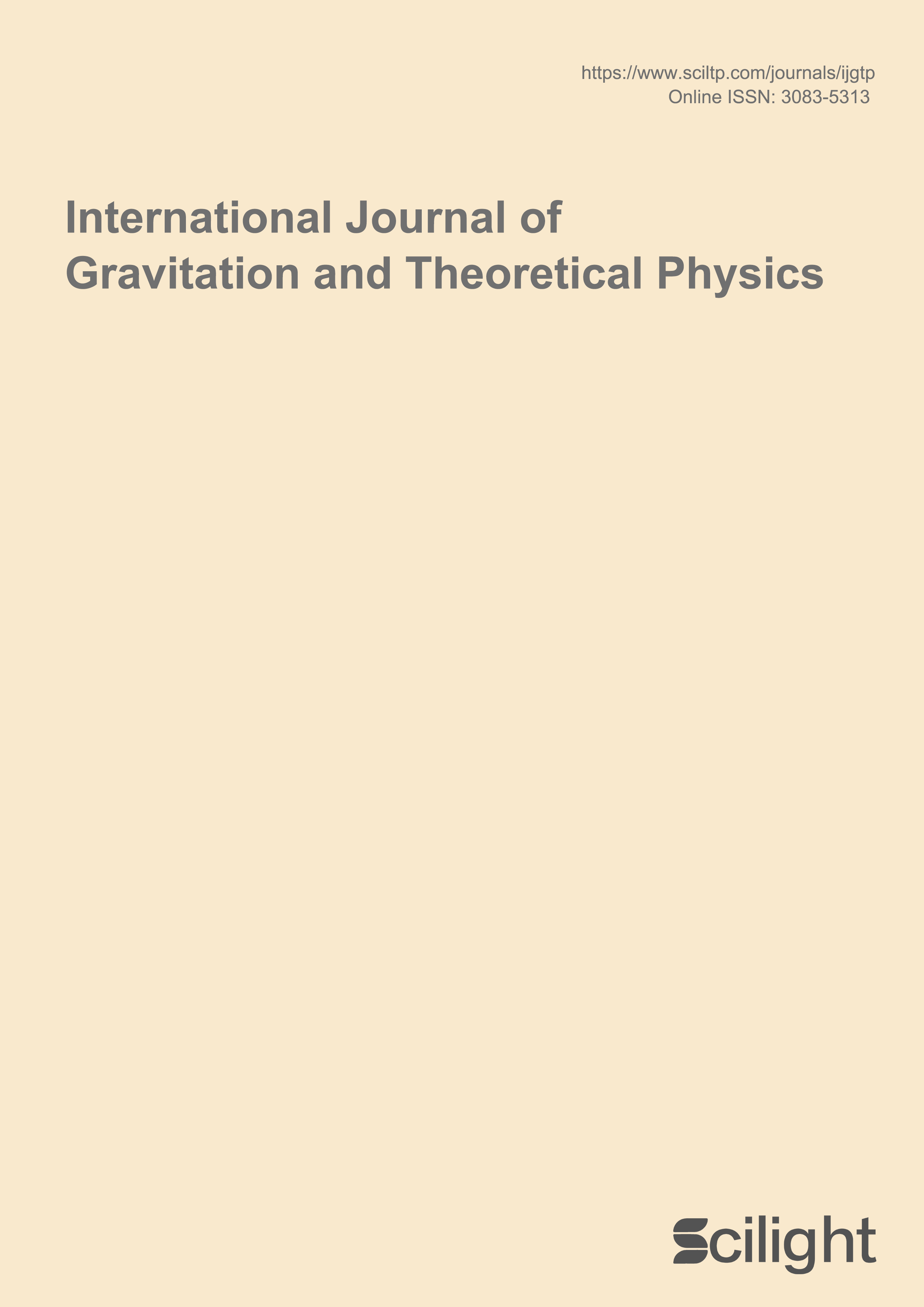} \vspace{-3pt}\\
        \end{tabular}
        \vspace{-22pt}
    \end{table}}
   
    \fancyfoot[C]{
        \vspace{-1.55cm}
        \begin{table}[H]
            \begin{minipage}[c]{0.15\columnwidth}
                \includegraphics[scale=0.5]{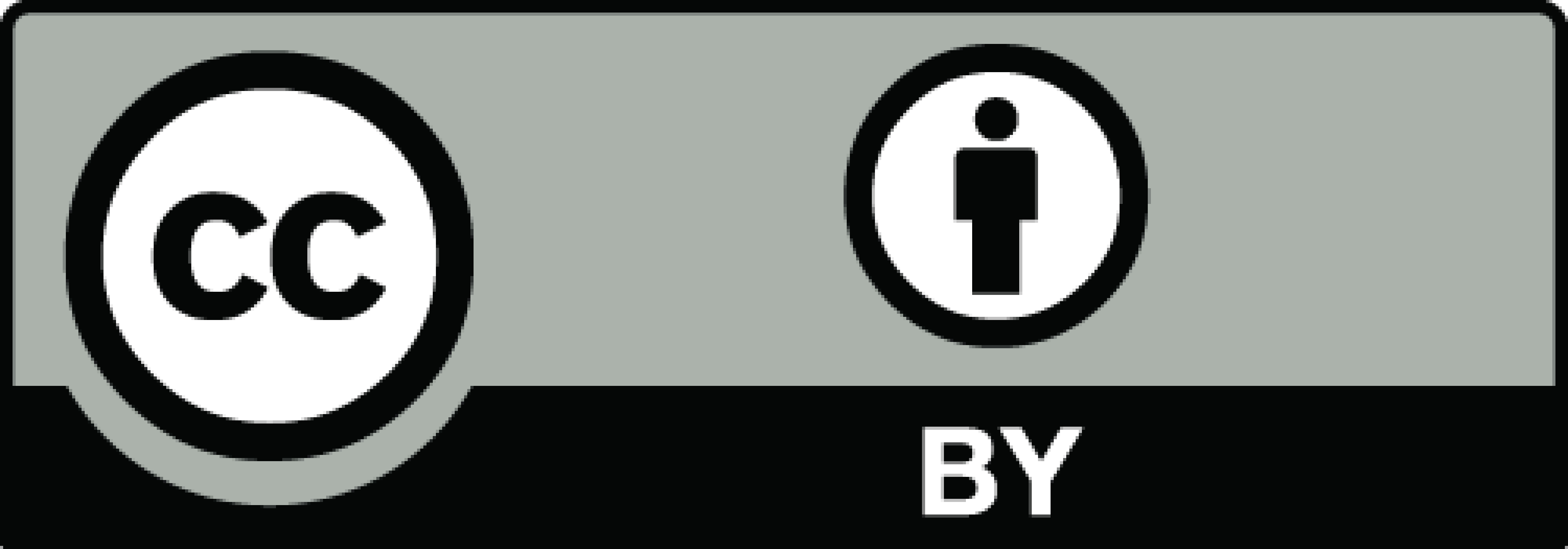} \vspace{1.1pt}
            \end{minipage}
            \hfill
            \begin{minipage}[c]{0.85\columnwidth}
                \scriptsize \textbf{Copyright:} © 2026 by the authors. This is an open access article under the terms and conditions of the Creative Commons Attribution (\mbox{CC BY}) license (\href{https://creativecommons.org/licenses/by/4.0/}{https://creativecommons.org/licenses/by/4.0/}). \\ \textbf{Publisher’s Note:} Scilight stays neutral with regard to jurisdictional claims in published maps and institutional affiliations.
            \end{minipage}
    \end{table}}
    \vspace{-0.55cm}
}

\begin{document}
\newgeometry{left=2.5cm, right=2.5cm, top=1.8cm, bottom=4cm}
	\thispagestyle{firstpage}
	{\noindent \textit{Article
}\hfill
\href{https://crossmark.crossref.org/dialog/?doi=10.53941/xxx.xxx.xxx\&domain=pdf}{
    \includegraphics[height=1.3em]{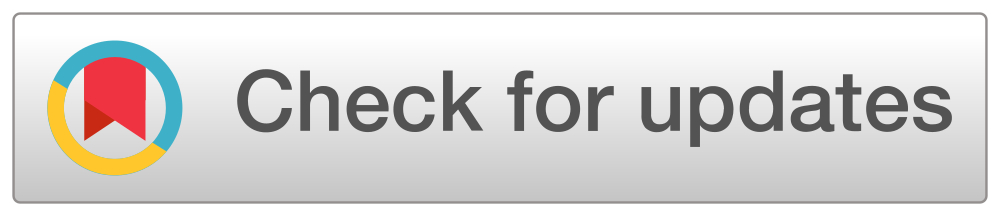} 
    \hspace*{-0.25cm} 
  }
}
	\vspace{4pt} \\
	{\fontsize{18pt}{10pt}\textbf{Mannheim--Kazanas Black Holes: Horizons, Temperatures and Thermodynamics}}
	\vspace{16pt} \\
	{\large Bekir Can L{\"u}tf{\"u}o{\u{g}}lu}
	\vspace{6pt}
	 \begin{spacing}{0.9}
		{\noindent \small
			Department of Physics, Faculty of Science, University of Hradec Kr{\'a}lov{\'e}, Rokitansk{\'e}ho 62/26, 500 03 Hradec Kr{\'a}lov{\'e}, Czech Republic; bekir.lutfuoglu@uhk.cz \vspace{6pt}
\\
		\footnotesize	\textbf{How To Cite}: L{\"u}tf{\"u}o{\u{g}}lu, B.C. Mannheim--Kazanas Black Holes: Horizons, Temperatures and Thermodynamics. \emph{International Journal of Gravitation and Theoretical Physics} \textbf{2026}, \emph{2}(3), 3. \href{https://doi.org/10.53941/ijgtp.2026.100017}{https://doi.org/10.53941/ijgtp.2026.100017}}
	\end{spacing}

\begin{table}[H]
\noindent\rule[0.15\baselineskip]{\textwidth}{0.5pt} 
\begin{tabular}{lp{12cm}} 
 \small 
  \begin{tabular}[t]{@{}l@{}} 
  \footnotesize  Received: 21 August 2026 \\
  \footnotesize  Revised: 17 September 2026\\
   \footnotesize Accepted: 20 September 2026 \\
  \footnotesize  Published: 24 September 2026
  \end{tabular} &
  \textbf{Abstract:} Static black holes in conformal Weyl gravity differ from their Einstein counterparts in a simple but important way: the Mannheim--Kazanas geometry is encoded in a Bach-flat lapse function whose radial dependence is richer than that of Schwarzschild--de Sitter.  We study the de Sitter branch of this solution in the Schwarzschild gauge, keeping the characteristic linear term together with the quadratic de Sitter term.  In the branch continuously connected to the positive-mass Schwarzschild--de Sitter spacetime, a regular static region between a black-hole horizon and a cosmological horizon exists precisely when 
 $-\frac{1}{3}<\beta\gamma<\frac{2}{3}$, and  $ 0<\kappa<\frac{1+3\beta\gamma}{27\beta^{2}}$. In this window, the singularity at the origin is hidden, the two positive horizons obey $0<\rb<3\beta<\rc$, and the limiting endpoint is the Nariai geometry.  The thermodynamic interpretation is subtler than the root structure: the two horizons generally have different temperatures, and the static patch has no asymptotic region that would select a unique normalization of time.  We therefore compare the Killing, Bousso--Hawking normalized, and Tolman local temperature conventions, and we compute the Wald entropy for the pure Weyl-squared action.  The result is a local horizon thermodynamic description that keeps the geometric facts separate from the ensemble-dependent choices needed for any global first law.\\
\\
  & 
\end{tabular}
\noindent\rule[0.15\baselineskip]{\textwidth}{0.5pt} 
\end{table}

\section{Introduction}

The Mannheim--Kazanas spacetime sits at the intersection of three well-developed lines of research.  The first is four-dimensional conformal Weyl gravity itself, whose local action is built from the square of the Weyl tensor and whose vacuum field equations are fourth order.  The theory has been studied both as a possible ultraviolet-complete or boundary-condition-restricted rewriting of gravity and as a phenomenological alternative to dark matter and dark energy \cite{Mannheim2006,Mannheim2012,Maldacena2011}.  The second line is the classification of static, spherically symmetric vacuum geometries in this theory.  Riegert established the conformal-gravity analogue of Birkhoff's theorem, and Mannheim and Kazanas displayed the form most often used in phenomenology: besides the familiar $1/r$ and $r^{2}$ terms, the lapse contains a term linear in $r$ \cite{Riegert1984,MannheimKazanas1989}.  In the notation used below, the de Sitter branch is written in terms of $\beta$, $\gamma$, and $\kappa$ after the Schwarzschild-gauge normalization has been fixed.

The third line is black-hole thermodynamics in spacetimes with more than one horizon.  Gibbons and Hawking showed that cosmological horizons have thermodynamic meaning, but a de Sitter black hole immediately raises a normalization and equilibrium problem: the static patch has no asymptotic infinity at which to normalize the timelike Killing vector, and the black-hole and cosmological horizons generally have different surface gravities \cite{GibbonsHawking1977}.  Bousso and Hawking proposed a useful normalization at the geodesic static radius, while later work on de Sitter black-hole thermodynamics has emphasized separate horizon first laws, effective temperatures, and the ensemble dependence of any combined thermodynamic description \cite{BoussoHawking1996,KubiznakSimovic2016}.  These issues are already present in Schwarzschild--de Sitter and cannot be avoided by adding the Mannheim--Kazanas linear lapse term.

\restoregeometry

The question of what should be called a temperature in a gravitational system has generated a wide range of viewpoints, from the standard thermodynamics of causal horizons to more ambitious proposals in which gravitational dynamics, screen variables, or even gravity itself are interpreted thermodynamically or entropically \cite{Jacobson1995,Padmanabhan2010,Verlinde2011,Konoplya2010Entropic}.  This range of perspectives is useful for the present de Sitter problem, where there is no preferred asymptotic clock and two different horizons appear in the same static patch.

A parallel thermodynamic literature exists for conformal-gravity black holes.  The key lesson is that the entropy is not the Einstein area entropy but the Wald Noether-charge entropy of the Weyl-squared action \cite{Wald1993,IyerWald1994}.  In anti-de Sitter conformal gravity, explicit first laws and phase structures have been worked out, including the role of additional spin-2 hair and, in later applications, black-hole chemistry and heat-engine variables \cite{LuPangPope2012,LiuLu2013,XuSunZhao2017}.  
This body of work makes clear that one must specify the action normalization, the boundary conditions, and the variables allowed to vary before assigning a global mass, pressure, or effective temperature.

\textls[-20]{Beyond thermodynamics, the Mannheim--Kazanas solution has also been used as a background for phenomenological probes.  Light bending in the non-asymptotically flat metric has been revisited in conformal Weyl gravity \cite{SultanaKazanas2010},} while perturbative studies have analyzed near-extremal geodesic instability, scalar, electromagnetic, and gravitational quasinormal modes, multi-stage ringdown and late-time tails, shadows, gray-body factors, and related Weyl black-hole/wormhole extensions \cite{MomenniaHendi2019,MomenniaHendi2020,MomenniaHendiBidgoli2021,Konoplya:2020fwg,Malik2024,Konoplya:2025mvj,Lutfuoglu:2025hjy,Kouniatalis:2025pxs}.  These citations are included only for orientation; none of those perturbative results is used in the derivations below.

The present paper is meant to fill a more elementary but important gap.  It does not propose a new solution and it does not attempt a full phase diagram.  The claim is therefore organizational and interpretive: the aim is to keep the de Sitter horizon structure, temperature normalizations, and Weyl-gravity entropy in one internally consistent convention.  Instead, it gives a self-contained classical analysis of the de Sitter Mannheim--Kazanas branch in which both the linear and quadratic lapse terms are retained: $\gamma\ne0$ and $\kappa\ne0$.  Concretely, the paper combines, in one place, four ingredients that are often treated separately: the precise parameter window for a positive-mass de Sitter black hole; the ordering and factorization of the three horizon roots; the comparison of Killing, Bousso--Hawking normalized, and Tolman local temperature conventions; and the corresponding Wald entropy in pure Weyl gravity.  The emphasis is therefore on consistency and interpretation.  The article identifies which statements are invariant horizon statements and which depend on the thermodynamic ensemble chosen for a spacetime with no static infinity.

Throughout we use units $c=\hbar=k_{\mathrm B}=G=1$ unless stated otherwise.  The symbol $\kappa$ denotes the coefficient of the $r^{2}$ term in the metric.  To avoid a clash of notation, the surface gravity of a horizon is denoted by $\varkappa$.

\section{Conformal-Gravity Conventions}

Consider the pure Weyl-squared action
\begin{equation}
    I_{\CG}=\frac{\TW}{16\pi}\int \dd^{4}x\,\sqrt{-g}\,
    C_{\mu\nu\rho\sigma}C^{\mu\nu\rho\sigma},
    \label{eq:weyl-action}
\end{equation}

In four spacetime dimensions the Weyl tensor is the trace-free part of the Riemann tensor,
\begin{equation}
    C_{\mu\nu\rho\sigma}
    =R_{\mu\nu\rho\sigma}
    -\frac{1}{2}\left(
      g_{\mu\rho}R_{\nu\sigma}-g_{\mu\sigma}R_{\nu\rho}
      -g_{\nu\rho}R_{\mu\sigma}+g_{\nu\sigma}R_{\mu\rho}
    \right)
    +\frac{R}{6}\left(g_{\mu\rho}g_{\nu\sigma}-g_{\mu\sigma}g_{\nu\rho}\right),
    \label{eq:weyl-tensor-definition}
\end{equation}
where $R_{\mu\nu\rho\sigma}$, $R_{\mu\nu}$, and $R$ are respectively the Riemann tensor, Ricci tensor, and Ricci scalar.  Equivalently, $C_{\mu\nu\rho\sigma}$ is the part of the curvature left after removing all Ricci traces; it vanishes for locally conformally flat metrics.

The remaining constant $\TW$ in \eqref{eq:weyl-action} is the dimensionless Weyl coupling in these units.  Some conformal-gravity papers use the opposite overall sign in the action.  All entropies below scale with this signed coupling; changing the sign convention for the action changes the sign convention for $\TW$.

The field equations obtained from \eqref{eq:weyl-action} are the vanishing of the Bach tensor in vacuum.  In the static, spherically symmetric gauge,
\begin{equation}
    \dd s^{2}=-B(r)\,\dd t^{2}+\frac{\dd r^{2}}{B(r)}+r^{2}\dd\Omega_{2}^{2},
    \label{eq:metric-ansatz}
\end{equation}

Here $t$ is the static time coordinate, $r$ is the areal radius, $\dd\Omega_{2}^{2}=\dd\theta^{2}+\sin^{2}\theta\,\dd\phi^{2}$ is the unit-two-sphere metric, and $B(r)$ is the metric lapse function.  This $B(r)$ should not be confused with the Bach tensor; no separate symbol for the Bach tensor is needed below.  A static region is an interval on which $B(r)>0$, since then $\partial_t$ is timelike.
The general Bach-flat lapse in this gauge contains constant, $1/r$, $r$, and $r^{2}$ pieces.  In the Mannheim--Kazanas parametrization used here it may be written as \cite{MannheimKazanas1989}.

\begin{equation}
    B(r)=1-3\beta\gamma-\frac{\beta(2-3\beta\gamma)}{r}
    +\gamma r-\kappa r^{2}.
    \label{eq:MK-lapse}
\end{equation}

Here $\beta$ has dimensions of length, $\gamma$ has dimensions of inverse length, and $\kappa$ has dimensions of inverse length squared.  The original Mannheim--Kazanas notation often writes the last parameter as $k$; we write $\kappa$ to make the de Sitter interpretation explicit.  It is useful to introduce
\begin{equation}
    A=1-3\beta\gamma,
    \qquad
    \mu=\beta(2-3\beta\gamma),
    \label{eq:A-mu-def}
\end{equation}
so that
\begin{equation}
    B(r)=A-\frac{\mu}{r}+\gamma r-\kappa r^{2}.
    \label{eq:B-Amu}
\end{equation}

The Schwarzschild--de Sitter limit is recovered by setting $\gamma=0$, in which case $\mu=2\beta$ and $\kappa=\Lambda/3$.
Here $\Lambda$ is the Einstein-gravity cosmological constant in the limiting metric.  The quantities $A$ and $\mu$ are only abbreviations.  This parametrization displays the four radial pieces of the Bach-flat lapse while using the conventional Mannheim--Kazanas normalization in which the de Sitter branch is labeled by $\beta$, $\gamma$, and $\kappa$.

\textls[-15]{The curvature invariants already show the role of the chosen parametrization.  For the lapse function \eqref{eq:MK-lapse}, one has}
\begin{align}
    C_{\mu\nu\rho\sigma}C^{\mu\nu\rho\sigma}
       &= \frac{12(\mu+\beta\gamma r)^{2}}{r^{6}},
       \label{eq:weyl-invariant}\\
    R &= 12\kappa-\frac{6\gamma}{r}+\frac{6\beta\gamma}{r^{2}}.
    \label{eq:ricci-scalar}
\end{align}

The Ricci scalar in \eqref{eq:ricci-scalar} follows from $R=-B''-4B'/r+2(1-B)/r^{2}$ for the metric \eqref{eq:metric-ansatz}.  The Weyl invariant follows from the single independent orthonormal Weyl component, $C_{\hat t\hat r\hat t\hat r}=-(\mu+\beta\gamma r)/r^{3}$, together with spherical symmetry; hats denote an orthonormal frame.
Thus the branch with $\mu>0$ contains a curvature singularity at $r=0$, and a black-hole interpretation requires an event horizon shielding that singularity from the static region.

\section{Horizon Structure and Parameter Range}

Horizons are positive zeros of $B(r)$.  Equivalently,
\begin{equation}
    rB(r)=-\kappa r^{3}+\gamma r^{2}+A r-\mu=0.
    \label{eq:horizon-cubic}
\end{equation}

For an asymptotically de Sitter branch we take
\begin{equation}
    \kappa>0,
    \label{eq:kappa-positive}
\end{equation}
so $B(r)\to -\infty$ as $r\to\infty$.  For the branch connected continuously to positive-mass Schwarzschild--de Sitter we also take
\begin{equation}
    \beta>0,
    \qquad
    \mu=\beta(2-3\beta\gamma)>0.
    \label{eq:positive-mu}
\end{equation}

The second condition means $\beta\gamma<2/3$ and ensures $B(r)\to -\infty$ as $r\to0^{+}$, so that the singularity is not already in the static region.

\begin{proposition}[Black-hole window]
Assume $\beta>0$ and $\kappa>0$.  The Mannheim--Kazanas solution \eqref{eq:MK-lapse} has the standard de Sitter black-hole horizon structure---one event horizon and one cosmological horizon, with the singularity hidden behind the event horizon---in the Schwarzschild--de Sitter-connected branch if
\begin{equation}
    -\frac{1}{3}<\beta\gamma<\frac{2}{3},
    \qquad
    0<\kappa<\kN,
    \qquad
    \kN=\frac{1+3\beta\gamma}{27\beta^{2}}.
    \label{eq:black-hole-window}
\end{equation}

In this range there are precisely two positive simple roots, denoted $\rb$ and $\rc$, and one negative root, denoted $\rn$, with
\begin{equation}
    0<\rb<3\beta<\rc,
    \qquad
    \rn<0.
    \label{eq:horizon-ordering}
\end{equation}

At $\kappa=\kN$ the two positive horizons merge at $r=3\beta$, giving the Nariai limit.
\end{proposition}

\begin{proof}
The coefficient of $1/r$ is $-\mu/r$.  If $\mu>0$, then $B(r)\to-\infty$ as $r\to0^{+}$.  Since $\kappa>0$, also $B(r)\to-\infty$ as $r\to\infty$.  Evaluating the lapse function at $r=3\beta$ gives

\begin{equation}
    B(3\beta)=\frac{1+3\beta\gamma}{3}-9\kappa\beta^{2}.
    \label{eq:B-at-3beta}
\end{equation}

Hence $B(3\beta)>0$ precisely when $\kappa<\kN$ and $1+3\beta\gamma>0$.  The inequalities in \eqref{eq:black-hole-window} combine this condition with $\mu>0$.  Since $B$ is negative at both ends of the positive radial axis but positive at $3\beta$, it has at least two positive zeros, one on each side of $3\beta$.  The cubic $rB(r)$ has at most three real roots.  Its constant term is $-\mu<0$ and its leading coefficient is $-\kappa<0$, so the product of all three roots is negative.  Once two roots are positive, the remaining real root is negative.  The equality $\kappa=\kN$ makes $B(3\beta)=B'(3\beta)=0$, so the two positive horizons coalesce at $r=3\beta$.
For completeness, the simultaneous equations $B(r_{*})=B'(r_{*})=0$ have candidate double roots $r_{*}=3\beta$ and $r_{*}=3\beta-2/\gamma$.  The second candidate is incompatible with the de Sitter black-hole window: for $\gamma>0$ it is positive only when $\beta\gamma>2/3$, while for $\gamma<0$ it would require $\kappa<0$.  Hence the only positive double-root endpoint in the present branch is the Nariai point.
\end{proof}

It is often useful to factor the cubic as
\begin{equation}
    rB(r)=-\kappa(r-\rb)(r-\rc)(r-\rn),
    \label{eq:factorized-cubic}
\end{equation}
where $\rb$ is the black-hole event horizon, $\rc$ is the cosmological horizon, and $\rn$ is negative in the black-hole window.  Comparing coefficients gives
\begin{align}
    \gamma &= \kappa(\rb+\rc+\rn),
    \label{eq:root-gamma}\\
    A &= -\kappa(\rb\rc+\rb\rn+\rc\rn),
    \label{eq:root-A}\\
    \mu &= -\kappa\rb\rc\rn.
    \label{eq:root-mu}
\end{align}

These relations are a compact way to translate between the Lagrangian-looking parameters $(\beta,\gamma,\kappa)$ and the directly geometric data $(\rb,\rc,\rn)$.

Several comments are worth making about the range \eqref{eq:black-hole-window}.  First, the assumption that both the linear and quadratic lapse terms are present means $\gamma\ne0$ and $\kappa\ne0$; the formulae above include the limiting cases only as checks.  Second, if $\beta\gamma\ge2/3$, then $\mu\le0$ and the $r=0$ singularity is not hidden in the Schwarzschild--de Sitter way.  Third, if $\beta\gamma\le-1/3$, the Nariai value $\kN$ is non-positive and the simple de Sitter black-hole interval disappears.  Other coordinate patches or conformal gauges may describe other causal structures, but they are not the positive-mass de Sitter black-hole branch analyzed here.  Thus ``black-hole window'' is used in a geometric sense: it asserts the root structure and shielding of the singularity, not perturbative stability or observational viability.

\subsection[The small-kappa and kappa=0 limits]{The Small-$\kappa$ and $\kappa=0$ Limits}

It is important not to confuse an outer horizon produced by the linear term with a genuinely de Sitter asymptotic.  In the sign convention of \eqref{eq:MK-lapse}, the invariant de Sitter scale is set by the coefficient of the quadratic term: as $r\to\infty$ one has $R\to12\kappa$ and the lapse behaves as $B(r)\sim-\kappa r^{2}$.  Thus a strictly de Sitter-like asymptotic requires $\kappa>0$, not a particular sign of $\gamma$.  The black-hole window \eqref{eq:black-hole-window} already includes either sign of $\gamma$, subject only to $-1/3<\beta\gamma<2/3$.

If $\kappa=0$ exactly, the spacetime leaves the de Sitter branch considered in the main analysis.  The large-$r$ curvature then tends to zero, since $R=-6\gamma/r+6\beta\gamma/r^{2}$ and $C_{\mu\nu\rho\sigma}C^{\mu\nu\rho\sigma}=12(\mu+\beta\gamma r)^{2}/r^{6}$.  Nevertheless the sign of $\gamma$ still matters for the number of horizons.  The horizon equation reduces to
\begin{equation}
    \gamma r^{2}+(1-3\beta\gamma)r-\beta(2-3\beta\gamma)=0.
    \label{eq:kappa-zero-horizon-quadratic}
\end{equation}

For $-1/3<\beta\gamma<0$, i.e., negative but not too negative $\gamma$, this quadratic has two positive roots.  The larger root is an outer horizon created by the negative linear term, not by de Sitter curvature.  For $0<\beta\gamma<2/3$, the same $\kappa=0$ equation has only one positive root, the black-hole horizon, and the static exterior extends to arbitrarily large $r$ with $B(r)\sim\gamma r$.  The borderline $\gamma=0=\kappa$ is the Schwarzschild limit, again with no cosmological horizon.

For very small but positive $\kappa$, both possibilities are continuously connected to the formulas above, but their interpretation differs.  If $\gamma>0$, the outer horizon is pushed to a large scale, approximately $\rc\sim\gamma/\kappa$, so it is the eventual $-\kappa r^{2}$ falloff that closes the static patch.  If $\gamma=0$, the familiar Schwarzschild--de Sitter scaling $\rc\sim1/\sqrt{\kappa}$ is recovered.  If $\gamma<0$, the two positive horizons approach the two finite roots of \eqref{eq:kappa-zero-horizon-quadratic} as $\kappa\to0^{+}$, while the third root runs to large negative radius.  This is why the present paper treats $\kappa=0$ only as a limiting check: it may have a two-horizon static patch for $\gamma<0$, but it is not asymptotically de Sitter in the curvature sense used for the thermodynamic discussion.

\section{Hawking Temperatures}

For a Killing horizon generated by $\chi=\partial_{t}$, the surface gravity associated with the coordinate normalization of $t$ is
\begin{equation}
    \varkappa_{\mathrm h}^{(K)}=\frac{1}{2}|B'(\rh)|,
    \qquad
    T_{\mathrm h}^{(K)}=\frac{|B'(\rh)|}{4\pi}.
    \label{eq:killing-temperature-general}
\end{equation}

The superscript $(K)$ is a reminder that this is the temperature associated with the unnormalized Killing vector $\partial_{t}$.  In de Sitter there is no infinity at which to normalize $\partial_{t}$, so this is a convenient geometric normalization.

Using the horizon equation $B(\rh)=0$, the derivative can be written as
\begin{equation}
    B'(\rh)=\frac{A}{\rh}+2\gamma-3\kappa\rh.
    \label{eq:Bprime-horizon}
\end{equation}

Therefore, in the black-hole window,
\begin{align}
    T_{\mathrm b}^{(K)}
       &=\frac{1}{4\pi}\left(\frac{A}{\rb}+2\gamma-3\kappa\rb\right),
       \label{eq:T-b-A}\\
    T_{\mathrm c}^{(K)}
       &=\frac{1}{4\pi}\left(3\kappa\rc-2\gamma-\frac{A}{\rc}\right).
       \label{eq:T-c-A}
\end{align}

The signs have been chosen so that both temperatures are positive: $B'(\rb)>0$ and $B'(\rc)<0$.

The factorized form \eqref{eq:factorized-cubic} gives an even clearer expression:
\begin{align}
    T_{\mathrm b}^{(K)}
       &=\frac{\kappa(\rc-\rb)(\rb-\rn)}{4\pi\rb},
       \label{eq:T-b-roots}\\
    T_{\mathrm c}^{(K)}
       &=\frac{\kappa(\rc-\rb)(\rc-\rn)}{4\pi\rc}.
       \label{eq:T-c-roots}
\end{align}

Since $\rn<0$ and $\rc>\rb$, these obey
\begin{equation}
    T_{\mathrm b}^{(K)}>T_{\mathrm c}^{(K)}
    \label{eq:Tb-greater-Tc}
\end{equation}

Indeed, the ratio of \eqref{eq:T-b-roots} and \eqref{eq:T-c-roots} is $\rc(\rb-\rn)/[\rb(\rc-\rn)]$, and the numerator exceeds the denominator by $(-\rn)(\rc-\rb)>0$; hence the inequality holds for every non-degenerate black hole in this branch.  The two horizons therefore do not form a global thermal equilibrium system.  A Euclidean continuation can make the geometry regular at the black-hole horizon or at the cosmological horizon, but not at both simultaneously, except in the limiting sense in which the two horizons coalesce.

\subsection{Bousso--Hawking Normalization}

A common de Sitter prescription is to normalize the Killing vector at the static geodesic radius, namely at the radius $\rzero$ where the acceleration of a static observer vanishes \cite{BoussoHawking1996}.  In the present metric this is the maximum of $B$ in the static patch:
\begin{equation}
    B'(\rzero)=0,
    \qquad
    2\kappa \rzero^{3}-\gamma\rzero^{2}-\mu=0,
    \qquad
    \rb<\rzero<\rc.
    \label{eq:rzero-def}
\end{equation}

Let
\begin{equation}
    B_{0}=B(\rzero)>0.
    \label{eq:B0-def}
\end{equation}

The normalized Killing vector is
\begin{equation}
    \widetilde{\chi}=\frac{\partial_{t}}{\sqrt{B_{0}}},
    \label{eq:normalized-killing}
\end{equation}
and the corresponding Bousso--Hawking temperatures are
\begin{equation}
    T_{\mathrm b}^{(BH)}=\frac{T_{\mathrm b}^{(K)}}{\sqrt{B_{0}}},
    \qquad
    T_{\mathrm c}^{(BH)}=\frac{T_{\mathrm c}^{(K)}}{\sqrt{B_{0}}}.
    \label{eq:bousso-hawking-temp}
\end{equation}

This normalization does not make the two horizons have the same temperature; it rescales both by the same factor.  It is best understood as the temperature measured by the preferred inertial static observer at $r=\rzero$.  Consequently, it is a normalization prescription.  The factor $B_{0}^{-1/2}$ depends on the integration constants of the solution.  The Bousso--Hawking temperature therefore cannot in general replace the Killing temperature in a first law while leaving the energy and conjugate quantities unchanged.  Such a change requires a consistent redefinition of the thermodynamic quantities, including the variation of the normalization factor.  The schematic first laws and restricted heat capacity below use the Killing temperature associated with $\partial_t$.

The Nariai limit illustrates the difference between normalizations.  As $\kappa\to\kN^{-}$, the horizons approach
\begin{equation}
    \rb,\rc\to 3\beta,
    \label{eq:nariai-radius}
\end{equation}
and the Killing temperatures \eqref{eq:T-b-A} and \eqref{eq:T-c-A} vanish.  However, $B_{0}$ also vanishes.  Taking the ratio gives a finite Bousso--Hawking temperature,
\begin{equation}
    T_{\mathrm N}^{(BH)}=\lim_{\kappa\to\kN^{-}}T_{\mathrm b}^{(BH)}
    =\lim_{\kappa\to\kN^{-}}T_{\mathrm c}^{(BH)}
    =\frac{1}{6\pi\beta}.
    \label{eq:nariai-bh-temperature}
\end{equation}

The last limit follows by writing $\rb=3\beta-\varepsilon+O(\varepsilon^{2})$ and $\rc=3\beta+\varepsilon+O(\varepsilon^{2})$ with small positive $\varepsilon$.  At the degenerate point $B''(3\beta)=-2/(9\beta^{2})$, so $T_{\mathrm b}^{(K)}\sim \varepsilon/(18\pi\beta^{2})$ while $\sqrt{B_{0}}\sim\varepsilon/(3\beta)$; their ratio is therefore $1/(6\pi\beta)$.
Thus the statement that the Nariai temperature is zero or non-zero is not a contradiction; it depends on whether one keeps the original Killing time or the rescaled time appropriate to the limiting Nariai geometry.

\subsection{Tolman Local Temperatures}

For an arbitrary static observer at radius $R$ with $\rb<R<\rc$, the locally measured Tolman temperature associated with a horizon is
\begin{equation}
    T_{\mathrm h}(R)=\frac{T_{\mathrm h}^{(K)}}{\sqrt{B(R)}}.
    \label{eq:tolman-temperature}
\end{equation}

Equation \eqref{eq:bousso-hawking-temp} is the special case $R=\rzero$.  The Tolman temperature diverges as the observer approaches a horizon, as expected from the infinite acceleration needed to remain static there.

\subsection{Temperature Dependence on Parameters}

The formulae above are most easily visualized in dimensionless variables.  Define
\begin{equation}
    x=\beta\gamma,
    \qquad
    q=\kappa\beta^{2},
    \qquad
    q_{\mathrm N}=\kN\beta^{2}=\frac{1+3x}{27}.
    \label{eq:dimensionless-plot-variables}
\end{equation}

Then the black-hole window is $-1/3<x<2/3$ and $0<q<q_{\mathrm N}$, while the dimensionless temperatures are $\beta T_{\mathrm b}$ and $\beta T_{\mathrm c}$.  Figure~\ref{fig:temperature-curves} shows representative one-parameter cuts through this window.  The left panel uses the Killing normalization; the right panel uses the Bousso--Hawking normalization at the geodesic radius.  In both panels the solid curves denote the black-hole horizon and the dashed curves denote the cosmological horizon.  The Killing temperatures vanish at the Nariai endpoint, whereas the Bousso--Hawking normalized temperatures approach the common finite value \eqref{eq:nariai-bh-temperature}.

\begin{figure}[H]
    \centering
    \includegraphics[width=0.92\textwidth]{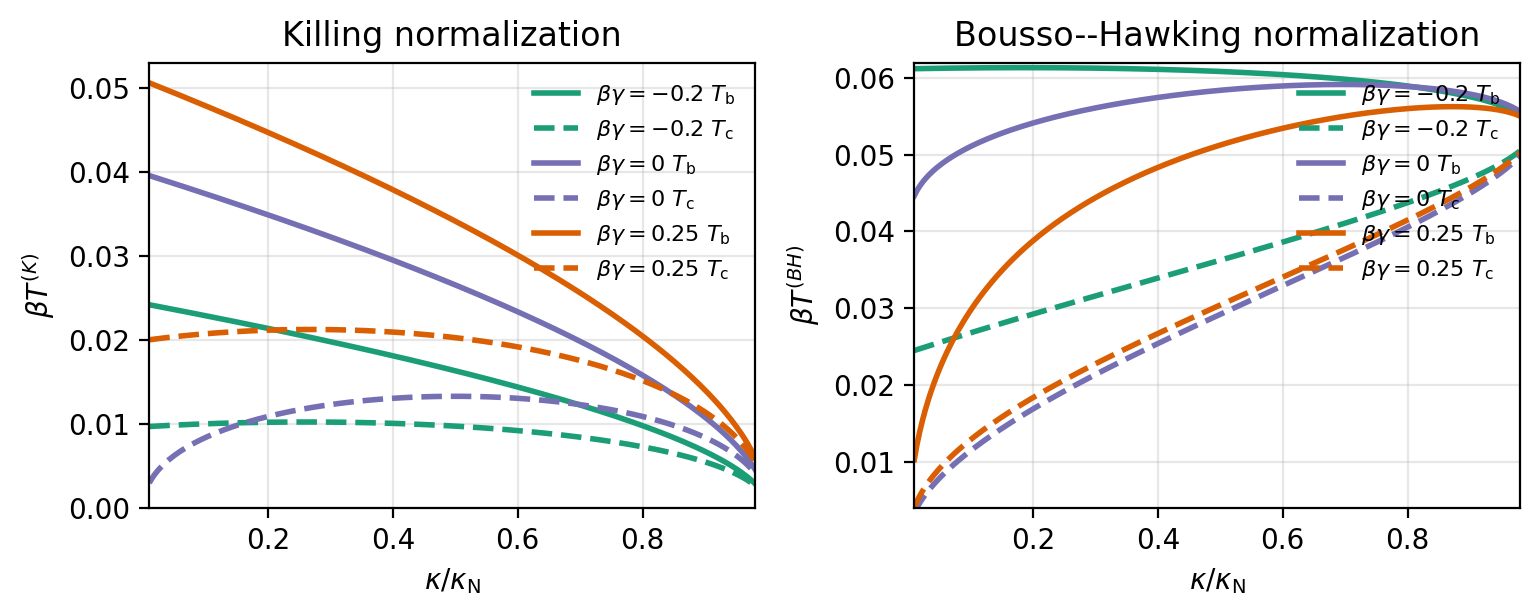}
    \caption{Dimensionless horizon temperatures as functions of $\kappa/\kappa_{\mathrm N}$ for three fixed values of $\beta\gamma$.  The solid curves are $\beta T_{\mathrm b}$ and the dashed curves are $\beta T_{\mathrm c}$.  The \textbf{left panel} shows the Killing temperatures; the \textbf{right panel} shows the Bousso--Hawking normalized temperatures.}
    \label{fig:temperature-curves}
\end{figure}

\vspace{-12pt}

Figure~\ref{fig:temperature-ratio-map} displays the temperature imbalance over the two-dimensional parameter space.  Since the Bousso--Hawking normalization rescales both horizon temperatures by the same factor, the ratio $T_{\mathrm b}/T_{\mathrm c}$ is the same for the Killing and Bousso--Hawking normalizations.  The plot confirms that the black-hole horizon is hotter throughout the non-degenerate branch and that the ratio tends to one only as the Nariai boundary is approached.

\begin{figure}[H]
    \centering
    \includegraphics[width=0.72\textwidth]{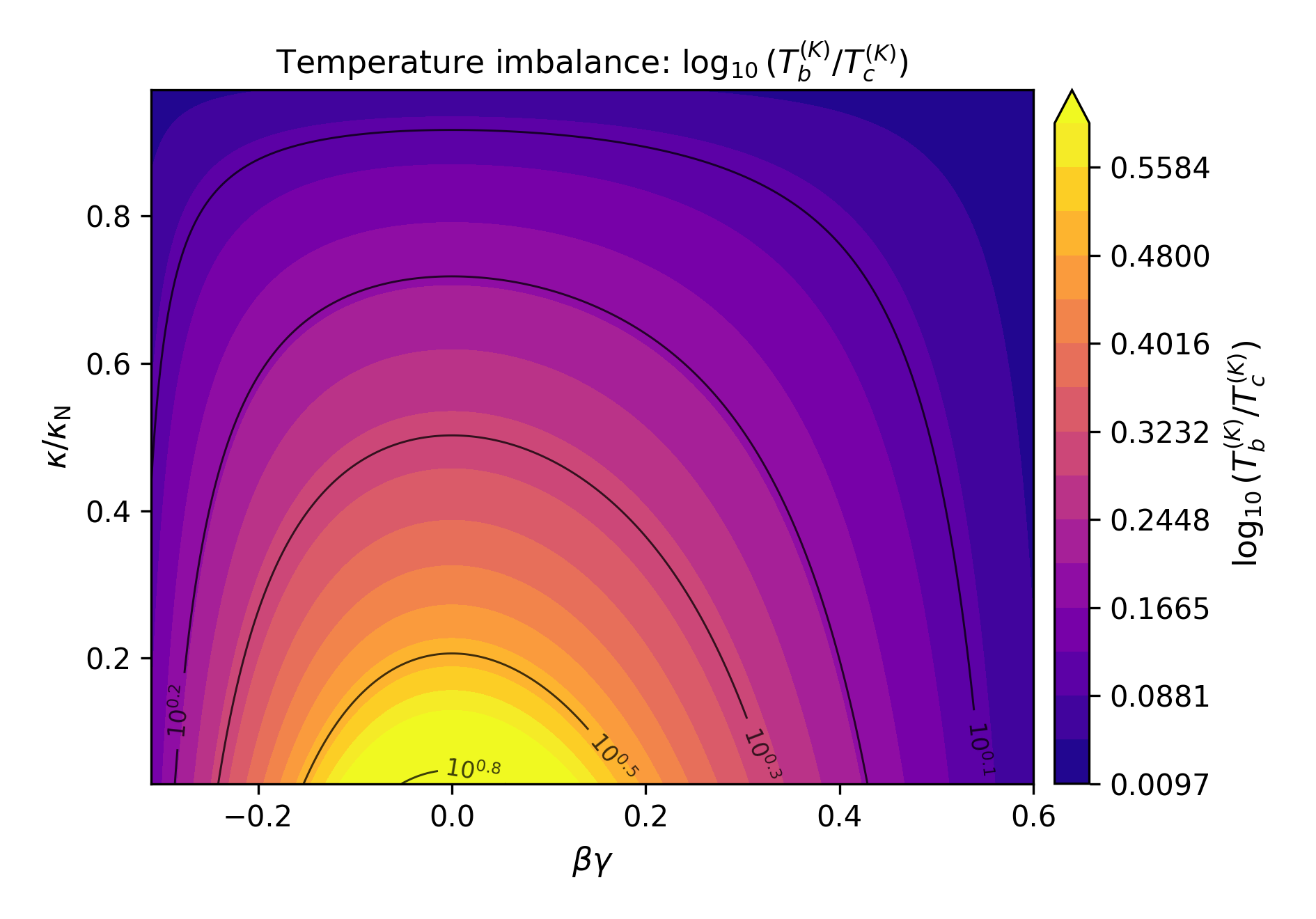}
    \caption{Parameter-space map of $\log_{10}(T_{\mathrm b}^{(K)}/T_{\mathrm c}^{(K)})$ in the black-hole window, using $x=\beta\gamma$ and $q/q_{\mathrm N}=\kappa/\kappa_{\mathrm N}$.  The ratio is independent of the common Bousso--Hawking redshift normalization.}
    \label{fig:temperature-ratio-map}
\end{figure}

\vspace{-24pt}

\subsection[Limitations of a combined thermodynamic description]{Limitations of a Combined Thermodynamic Description}

One sometimes introduces an effective temperature for the entire de Sitter static patch by combining the event horizon and the cosmological horizon into a single thermodynamic system.  Such constructions are useful, but they are not unique.  Their form depends on which quantity is called the internal energy, whether the cosmological scale is varied as pressure, whether the thermodynamic volume is the inter-horizon volume, and whether the entropy is taken to be a sum of the two horizon entropies, possibly with correlation terms.

For the Mannheim--Kazanas spacetime the same caveat is stronger.  In conformal gravity the coefficient $\gamma$ of the linear $r$ term in the lapse is an integration constant of the Bach-flat solution, and a complete first law generally contains additional work terms.  Therefore an effective temperature should be regarded as an ensemble-dependent thermodynamic variable, not as the Hawking temperature of either horizon.  The calculations here concern the separate horizon temperatures \eqref{eq:T-b-A} and \eqref{eq:T-c-A} and their normalized or local versions; no effective two-horizon temperature is defined or calculated.

\section{Wald Entropy and Horizon Thermodynamics}

Because the action is quadratic in curvature, the entropy is not proportional to area.  For the action \eqref{eq:weyl-action}, the Wald entropy of a Killing horizon $\mathcal H$ is
\begin{equation}
    S_{\mathcal H}
    =-2\pi\int_{\mathcal H}\dd^{2}x\,\sqrt{h}\,
      \frac{\partial \mathcal L}{\partial R_{\mu\nu\rho\sigma}}
      \epsilon_{\mu\nu}\epsilon_{\rho\sigma},
    \label{eq:wald-definition}
\end{equation}
where $\mathcal L=(\TW/16\pi)C_{\alpha\beta\lambda\delta}C^{\alpha\beta\lambda\delta}$ and $\epsilon_{\mu\nu}$ is the binormal to the horizon cross-section.  For the metric \eqref{eq:metric-ansatz},
\begin{equation}
    C_{\hat t\hat r\hat t\hat r}
    =\frac{r^{2}B''-2rB'+2B-2}{6r^{2}},
    \label{eq:weyl-orthonormal}
\end{equation}
where hatted indices denote an orthonormal frame.  
For the Lagrangian density in \eqref{eq:wald-definition}, the trace-free property of the Weyl tensor gives $\partial\mathcal L/\partial R_{\mu\nu\rho\sigma}=(\TW/8\pi)C^{\mu\nu\rho\sigma}$ inside the binormal contraction.  With the orientation convention used here this reduces, for a spherical horizon of radius $\rh$, to $S_{\mathrm h}=-4\pi\TW\rh^{2}C_{\hat t\hat r\hat t\hat r}|_{r=\rh}$.  Since \eqref{eq:MK-lapse} gives $C_{\hat t\hat r\hat t\hat r}=-(\mu+\beta\gamma r)/r^{3}$, the entropy below follows directly.
Substituting \eqref{eq:MK-lapse} and evaluating on a horizon gives

\begin{equation}
    S_{\mathrm h}
    =4\pi\TW\left(\frac{\mu}{\rh}+\beta\gamma\right).
    \label{eq:wald-entropy-mu}
\end{equation}

Using $B(\rh)=0$, the same entropy can be written as
\begin{equation}
    S_{\mathrm h}
    =4\pi\TW\left(1-2\beta\gamma+\gamma\rh-\kappa\rh^{2}\right).
    \label{eq:wald-entropy-rh}
\end{equation}

This formula applies to either $\rh=\rb$ or $\rh=\rc$, with the sign fixed by the action convention in \eqref{eq:weyl-action}.  If a topological Gauss--Bonnet term is added to the action, the entropy is shifted by a constant; the horizon dependence in \eqref{eq:wald-entropy-rh} is unchanged.

Several checks are immediate.  In the Schwarzschild--de Sitter limit $\gamma=0$,
\begin{equation}
    S_{\mathrm h}=4\pi\TW(1-\kappa\rh^{2}).
    \label{eq:entropy-sds}
\end{equation}

For the black-hole horizon this decreases from $4\pi\TW$ at small $\kappa\beta^{2}$ to $8\pi\TW/3$ at the Nariai point.  For pure de Sitter, $\beta=\gamma=0$ and $\rh=1/\sqrt{\kappa}$, the Weyl tensor vanishes and the entropy from a pure Weyl-squared action is zero.  At the Mannheim--Kazanas Nariai point, $\rh=3\beta$, and \eqref{eq:wald-entropy-mu} gives
\begin{equation}
    S_{\mathrm N}=\frac{8\pi\TW}{3},
    \label{eq:nariai-entropy}
\end{equation}
independent of $\gamma$.

For the black-hole horizon, $S_{\mathrm b}/(4\pi\TW)=\mu/\rb+\beta\gamma>\mu/(3\beta)+\beta\gamma=2/3$ throughout the non-degenerate window.  Thus $\TW>0$ is the entropy-positive choice for this horizon with the action convention of \eqref{eq:weyl-action}.  We adopt this sign in Figure~\ref{fig:heat-capacity-phase-transition}; reversing $\TW$ reverses both the entropy and the heat capacity.

Equations \eqref{eq:T-b-A}, \eqref{eq:T-c-A}, and \eqref{eq:wald-entropy-rh} are the basic thermodynamic data.  They are local to each horizon.  In the following schematic first laws we fix the generator to $\chi=\partial_t$ and use $T_{\mathrm b,c}^{(K)}$; the energies and conjugate quantities are to be defined consistently with this choice.  A black-hole horizon first law can then be written schematically as
\begin{equation}
    \dd E_{\mathrm b}
    =T_{\mathrm b}^{(K)}\,\dd S_{\mathrm b}
     +\Psi_{\gamma}^{(\mathrm b)}\,\dd\gamma
     +\Theta_{\kappa}^{(\mathrm b)}\,\dd\kappa+\cdots,
    \label{eq:first-law-b-schematic}
\end{equation}
while a cosmological-horizon first law has the opposite orientation for the horizon-energy term,

\begin{equation}
    \dd E_{\mathrm c}
    =-T_{\mathrm c}^{(K)}\,\dd S_{\mathrm c}
     +\Psi_{\gamma}^{(\mathrm c)}\,\dd\gamma
     +\Theta_{\kappa}^{(\mathrm c)}\,\dd\kappa+\cdots.
    \label{eq:first-law-c-schematic}
\end{equation}

Here $E_{\mathrm b}$ and $E_{\mathrm c}$ denote the horizon energies appropriate to the chosen ensemble; $\Psi_{\gamma}^{(\mathrm b,c)}$ and $\Theta_{\kappa}^{(\mathrm b,c)}$ are the quantities conjugate to variations of the integration constants $\gamma$ and $\kappa$, which set the linear and quadratic lapse terms, respectively.  The sign difference between \eqref{eq:first-law-b-schematic} and \eqref{eq:first-law-c-schematic} is the usual orientation difference between black-hole and cosmological horizons.
The ellipses emphasize that a fully specified first law requires a choice of boundary conditions and thermodynamic ensemble.  In anti-de Sitter conformal gravity, analogous first laws include an additional pair conjugate to the massive spin-2 hair; in the present de Sitter setting the linear Mannheim--Kazanas parameter plays a similar role as an independent integration constant.  Holding $\gamma$ and $\kappa$ fixed and varying only the horizon position gives a useful restricted horizon thermodynamics, but it should not be confused with a universal global mass formula.

For reference, if $\gamma$ and $\kappa$ are held fixed, the black-hole horizon equation may be solved for $\beta$ as a function of the black-hole radius:
\begin{equation}
    3\gamma\beta^{2}-(2+3\gamma\rb)\beta
    +\rb(1+\gamma\rb-\kappa\rb^{2})=0.
    \label{eq:beta-quadratic}
\end{equation}

The branch continuous at $\gamma=0$ is
\begin{equation}
    \beta(\rb)=
    \frac{2+3\gamma\rb-
    \sqrt{4-3\gamma^{2}\rb^{2}+12\gamma\kappa\rb^{3}}}{6\gamma},
    \qquad (\gamma\ne0),
    \label{eq:beta-rb-branch}
\end{equation}
with the smooth limit
\begin{equation}
    \beta(\rb)\longrightarrow \frac{\rb}{2}(1-\kappa\rb^{2})
    \qquad\text{as}\qquad \gamma\to0.
    \label{eq:beta-gamma-zero-limit}
\end{equation}

The square-root branch in \eqref{eq:beta-rb-branch} is the one with a finite Schwarzschild--de Sitter limit; the other algebraic branch is not used in the restricted black-hole response function below.  With $\TW$ fixed, a restricted fixed-$(\gamma,\kappa)$ heat capacity for the black-hole horizon, based on $T_{\mathrm b}^{(K)}$, can then be defined by

\begin{equation}
    C_{\mathrm b}^{(\gamma,\kappa)}
    =T_{\mathrm b}^{(K)}\left(\frac{\partial S_{\mathrm b}}{\partial T_{\mathrm b}^{(K)}}\right)_{\gamma,\kappa}
    =T_{\mathrm b}^{(K)}\,
      \frac{\dd S_{\mathrm b}(\rb)/\dd\rb}
           {\dd T_{\mathrm b}^{(K)}(\rb)/\dd\rb},
    \label{eq:heat-capacity}
\end{equation}
where $S_{\mathrm b}(\rb)$ and $T_{\mathrm b}^{(K)}(\rb)$ are obtained by setting $\rh=\rb$ in \eqref{eq:wald-entropy-rh} and using \eqref{eq:T-b-A}, and then inserting \eqref{eq:beta-rb-branch}.  The superscript reminds us that this is not a grand-canonical response function: only the black-hole horizon radius, equivalently $\beta$ through \eqref{eq:beta-rb-branch}, is varied.  If a different ensemble is chosen, extra terms involving $\dd\gamma$ or $\dd\kappa$ must be included before a heat capacity is assigned.  The sign of this heat capacity is an ensemble-dependent local-stability diagnostic, not a statement of global equilibrium for the whole static patch.

Figure~\ref{fig:heat-capacity-phase-transition} illustrates both the Killing temperature and the restricted response function defined in \eqref{eq:heat-capacity}.  With the entropy-positive sign convention $\alpha_{\mathrm W}>0$, the plotted black-hole heat capacity is positive over the representative branches.  The negative response obtained with the opposite sign of $\alpha_{\mathrm W}$ does not correspond to positive black-hole entropy for the action used here.  In these cuts the Killing temperature decreases monotonically along the branch and no pole appears in $C_{\mathrm b}^{(\gamma,\kappa)}$.  The heat capacity instead develops a finite positive maximum before decreasing to zero at the Nariai endpoint.

\begin{figure}[H]
    \centering
    \includegraphics[width=0.92\textwidth]{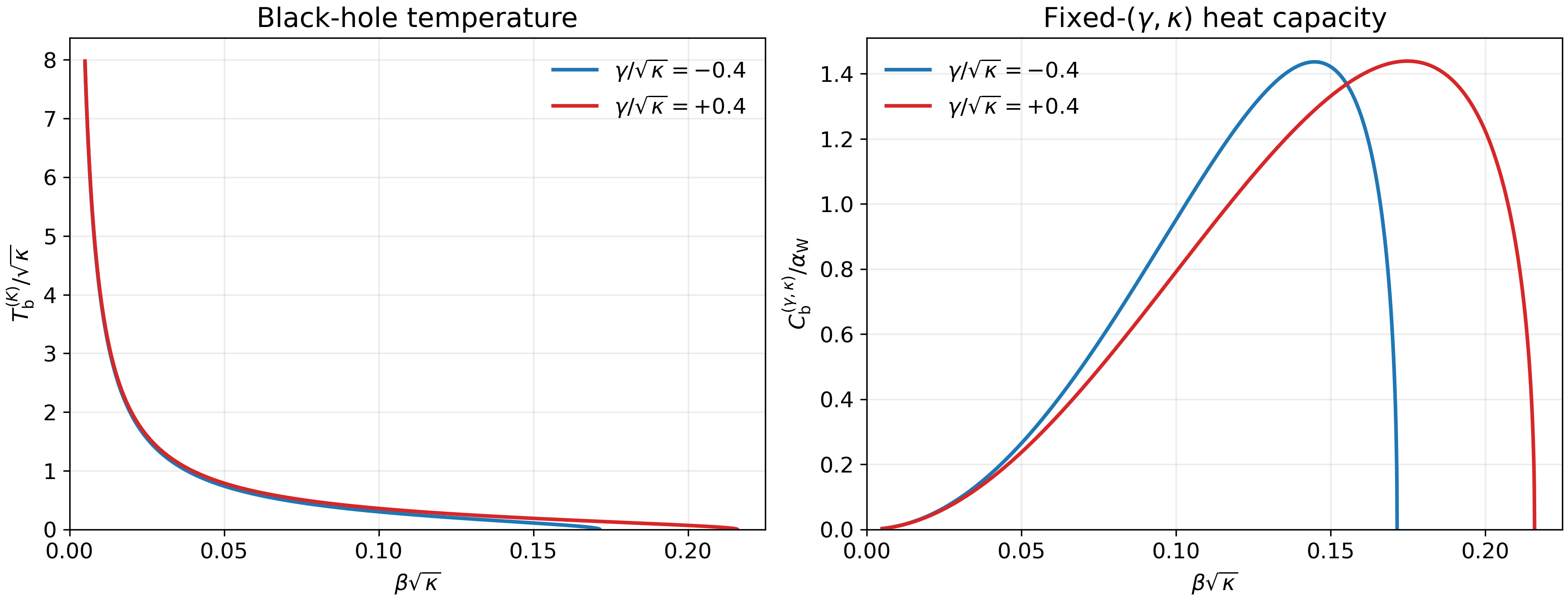}
    \caption{Representative fixed-$(\gamma,\kappa)$ temperature and response curves for the black-hole horizon.  The horizontal axis in both panels is the dimensionless branch parameter $\beta\sqrt{\kappa}$.  The \textbf{left panel} shows the Killing temperature $T_{\mathrm b}^{(K)}/\sqrt{\kappa}$ for the two cuts $\gamma/\sqrt{\kappa}=\pm0.4$.  The \textbf{right panel} overlays the corresponding fixed-$(\gamma,\kappa)$ heat capacities in different colors, plotting $C_{\mathrm b}^{(\gamma,\kappa)}/\alpha_{\mathrm W}$ for the entropy-positive choice $\alpha_{\mathrm W}>0$.  The response is defined using the Killing temperature in \eqref{eq:heat-capacity}.}
    \label{fig:heat-capacity-phase-transition}
\end{figure}

\vspace{-14pt}

This response-function language should be read with the same qualifications as the heat capacity itself.  The finite maxima visible in Figure~\ref{fig:heat-capacity-phase-transition} are local features of the fixed-$(\gamma,\kappa)$ equation of state, not evidence for a global equilibrium transition of the full de Sitter static patch.  The black-hole and cosmological horizons are generally at different temperatures, and no asymptotic reservoir fixes a unique energy.  A solution-independent positive rescaling of the temperature leaves $T\,\dd S/\dd T$ unchanged.  A solution-dependent rescaling, such as the Bousso--Hawking prescription, changes the temperature derivative and hence the response function; using it in a first law also requires consistent energy and conjugate quantities.  Allowing $\gamma$ or $\kappa$ to vary defines a different response problem.  Thus the robust statement is the finite, positive Killing-temperature response on the plotted fixed-$(\gamma,\kappa)$ branches for $\TW>0$, not a universal phase boundary of Mannheim--Kazanas de Sitter black holes.

\section{Summary}

The Mannheim--Kazanas metric with non-zero $\gamma$ and non-zero $\kappa$ has a clean de Sitter black-hole interpretation in the branch
\begin{equation}
    \beta>0,
    \qquad
    -\frac{1}{3}<\beta\gamma<\frac{2}{3},
    \qquad
    0<\kappa<\frac{1+3\beta\gamma}{27\beta^{2}}.
    \label{eq:summary-window}
\end{equation}

In that window the geometry has a singularity at $r=0$, an event horizon $\rb$, a cosmological horizon $\rc$, and a negative third root of the horizon cubic.  The positive horizons obey $0<\rb<3\beta<\rc$, and the upper endpoint of the parameter interval is the Nariai limit $\rb=\rc=3\beta$.

The horizon temperatures are not unique until one specifies the normalization of the Killing vector.  With the coordinate Killing vector $\partial_{t}$,

\begin{equation}
    T_{\mathrm b}^{(K)}
       =\frac{1}{4\pi}\left(\frac{1-3\beta\gamma}{\rb}+2\gamma-3\kappa\rb\right),
    \qquad
    T_{\mathrm c}^{(K)}
       =\frac{1}{4\pi}\left(3\kappa\rc-2\gamma-\frac{1-3\beta\gamma}{\rc}\right).
    \label{eq:summary-temperatures}
\end{equation}

They are unequal for every non-degenerate black hole in this branch.  Bousso--Hawking normalization rescales both by $1/\sqrt{B(\rzero)}$, where $B'(\rzero)=0$, and gives the finite Nariai value $1/(6\pi\beta)$.  Tolman temperatures add the usual redshift factor for an observer at a specified radius.  A combined thermodynamic description would require an additional ensemble prescription; no effective two-horizon temperature has been calculated here.

For the Weyl-squared action \eqref{eq:weyl-action}, the entropy of either horizon is
\begin{equation}
    S_{\mathrm h}=4\pi\TW\left(\frac{\beta(2-3\beta\gamma)}{\rh}+\beta\gamma\right)
    =4\pi\TW\left(1-2\beta\gamma+\gamma\rh-\kappa\rh^{2}\right).
    \label{eq:summary-entropy}
\end{equation}

This replaces the area law.  Together with the appropriate choice of horizon temperature, it provides a consistent local thermodynamic description.  A global first law requires additional input: boundary conditions, a definition of energy in the de Sitter static patch, and a decision about whether the integration constants $\gamma$ and $\kappa$ are held fixed or varied as thermodynamic variables with conjugate quantities.

		\section*{Funding}
        
This work was supported by the Excellence project of the Faculty of Science, University of Hradec Králové, grant number 2205/2025--2026.

 
		\section*{Institutional Review Board Statement}
Not applicable.

		\section*{Informed Consent Statement}
Not applicable.


		\section*{Data Availability Statement}
Not applicable.
 
		\section*{Acknowledgments}
B. C. L. is grateful to the Excellence project FoS UHK 2205/2025--2026 for the financial support.

\section*{Conflicts of Interest}
The authors declare no conflict of interest.




\section*{Use of AI and AI-Assisted Technologies}

No AI tools were utilized for this paper.




\small
	\begin{thebibliography}{99}	
	\bibitem{Mannheim2006}
Mannheim, P.D. Alternatives to dark matter and dark energy. \emph{Prog. Part. Nucl. Phys.} \textbf{2006}, \emph{56}, 340--445. \href{https://doi.org/10.1016/j.ppnp.2005.08.001}{https://doi.org/10.1016/j.ppnp.2005.08.001}.

\bibitem{Mannheim2012}
Mannheim, P.D. Making the case for conformal gravity. \emph{Found. Phys.} \textbf{2012}, \emph{42}, 388--420. \href{https://doi.org/10.1007/s10701-011-9608-6}{https://doi.org/10.1007/s10701-011-9608-6}.

\bibitem{Maldacena2011}
Maldacena, J. Einstein gravity from conformal gravity. \emph{arXiv} \textbf{2011}, arXiv:1105.5632.
		
\bibitem{Riegert1984}
Riegert, R.J. Birkhoff's theorem in conformal gravity. \emph{Phys. Rev. Lett.} \textbf{1984}, \emph{53}, 315--318. \url{https://doi.org/10.1103/physrevlett.53.315}.

\bibitem{MannheimKazanas1989}
Mannheim, P.D.; Kazanas, D. Exact vacuum solution to conformal Weyl gravity and galactic rotation curves. \emph{Astrophys. J.} \textbf{1989}, \emph{342}, 635. \href{https://doi.org/10.1086/167623}{https://doi.org/10.1086/167623}.

\bibitem{GibbonsHawking1977}
Gibbons, G.W.; Hawking, S.W. Cosmological event horizons, thermodynamics, and particle creation. \emph{Phys. Rev. D} \textbf{1977}, \emph{15}, 2738--2751. \href{https://doi.org/10.1103/physrevd.15.2738}{https://doi.org/10.1103/physrevd.15.2738}.

\bibitem{BoussoHawking1996}
Bousso, R.; Hawking, S.W. Pair creation of black holes during inflation. \emph{Phys. Rev. D} \textbf{1996}, \emph{54}, 6312--6322. \href{https://doi.org/10.1103/physrevd.54.6312}{https://doi.org/10.1103/physrevd.54.6312}.

\bibitem{KubiznakSimovic2016}
Kubizňák, D.; Simovic, F. Thermodynamics of horizons: de Sitter black holes and reentrant phase transitions. \emph{Class. Quantum Grav.} \textbf{2016}, \emph{33}, 245001. \href{https://doi.org/10.1088/0264-9381/33/24/245001}{https://doi.org/10.1088/0264-9381/33/24/245001}.

\bibitem{Jacobson1995}
Jacobson, T. Thermodynamics of spacetime: The Einstein equation of state. \emph{Phys. Rev. Lett.} \textbf{1995}, \emph{75}, 1260--1263. \href{https://doi.org/10.1103/physrevlett.75.1260}{https://doi.org/10.1103/physrevlett.75.1260}.

\bibitem{Padmanabhan2010}
Padmanabhan, T. Thermodynamical aspects of gravity: New insights. \emph{Rep. Prog. Phys.} \textbf{2010}, \emph{73}, 046901. \href{https://doi.org/10.1088/0034-4885/73/4/046901}{https://doi.org/10.1088/0034-4885/73/4/046901}.

\bibitem{Verlinde2011}
Verlinde, E. On the origin of gravity and the laws of Newton. \emph{J. High Energy Phys.} \textbf{2011}, \emph{2011}, 029. \href{https://doi.org/10.1007/jhep04(2011)029}{https://doi.org/10.1007/jhep04(2011)029}.

\bibitem{Konoplya2010Entropic}
Konoplya, R.A. Entropic force, holography and thermodynamics for static space-times. \emph{Eur. Phys. J. C} \textbf{2010}, \emph{69}, 555--562. \href{https://doi.org/10.1140/epjc/s10052-010-1424-1}{https://doi.org/10.1140/epjc/s10052-010-1424-1}.

\bibitem{Wald1993}
Wald, R.M. Black hole entropy is the Noether charge. \emph{Phys. Rev. D} \textbf{1993}, \emph{48}, R3427--R3431. \href{https://doi.org/10.1103/physrevd.48.r3427}{https://doi.org/10.1103/physrevd.48.r3427}.

\bibitem{IyerWald1994}
Iyer, V.; Wald, R.M. Some properties of the Noether charge and a proposal for dynamical black hole entropy. \emph{Phys. Rev. D} \textbf{1994}, \emph{50}, 846--864. \href{https://doi.org/10.1103/physrevd.50.846}{https://doi.org/10.1103/physrevd.50.846}.

\bibitem{LuPangPope2012}
Lü, H.; Pang, Y.; Pope, C.N.; et al. AdS and Lifshitz black holes in conformal and Einstein-Weyl gravities. \emph{Phys. Rev. D} \textbf{2012}, \emph{86}, 044011. \href{https://doi.org/10.1103/physrevd.86.044011}{https://doi.org/10.1103/physrevd.86.044011}.

\bibitem{LiuLu2013}
Liu, H.S.; Lü, H. Charged rotating AdS black hole and its thermodynamics in conformal gravity. \emph{J. High Energy Phys.} \textbf{2013}, \emph{2013}, 139. \href{https://doi.org/10.1007/jhep02(2013)139}{https://doi.org/10.1007/jhep02(2013)139}.

\bibitem{XuSunZhao2017}
Xu, H.; Sun, Y.; Zhao, L. Black hole thermodynamics and heat engines in conformal gravity. \emph{Int. J. Mod. Phys. D} \textbf{2017}, \emph{26}, 1750151. \href{https://doi.org/10.1142/s0218271817501516}{https://doi.org/10.1142/s0218271817501516}.

\bibitem{SultanaKazanas2010}
Sultana, J.; Kazanas, D. Bending of light in conformal Weyl gravity. \emph{Phys. Rev. D} \textbf{2010}, \emph{81}, 127502. \href{https://doi.org/10.1103/physrevd.81.127502}{https://doi.org/10.1103/physrevd.81.127502}.

\bibitem{MomenniaHendi2019}
Momennia, M.; Hendi, S.H. Near-extremal black holes in Weyl gravity: Quasinormal modes and geodesic instability. \emph{Phys. Rev. D} \textbf{2019}, \emph{99}, 124025. \href{https://doi.org/10.1103/physrevd.99.124025}{https://doi.org/10.1103/physrevd.99.124025}.

\bibitem{MomenniaHendi2020}
Momennia, M.; Hendi, S.H. Quasinormal modes of black holes in Weyl gravity: Electromagnetic and gravitational perturbations. \emph{Eur. Phys. J. C} \textbf{2020}, \emph{80}, 505. \href{https://doi.org/10.1140/epjc/s10052-020-8051-2}{https://doi.org/10.1140/epjc/s10052-020-8051-2}.

\bibitem{MomenniaHendiBidgoli2021}
Momennia, M.; Hendi, S.H.; Soltani Bidgoli, F. Stability and quasinormal modes of black holes in conformal Weyl gravity. \emph{Phys. Lett. B} \textbf{2021}, \emph{813}, 136028. \href{https://doi.org/10.1016/j.physletb.2020.136028}{https://doi.org/10.1016/j.physletb.2020.136028}.

\bibitem{Konoplya:2020fwg}
Konoplya, R.A. Conformal Weyl gravity via two stages of quasinormal ringing and late-time behavior. \emph{Phys. Rev. D} \textbf{2021}, \emph{103}, 044033. \href{https://doi.org/10.1103/physrevd.103.044033}{https://doi.org/10.1103/physrevd.103.044033}.

\bibitem{Malik2024}
Malik, Z. Quasinormal modes of the Mannheim--Kazanas black holes. \emph{Z. Naturforsch. A} \textbf{2024}, \emph{79}, 1063--1073. \href{https://doi.org/10.1515/zna-2024-0153}{https://doi.org/10.1515/zna-2024-0153}.

\bibitem{Konoplya:2025mvj}
Konoplya, R.A.; Khrabustovskyi, A.; Kříž, J.; et al. Quasinormal ringing and shadows of black holes and wormholes in dark matter-inspired Weyl gravity. \emph{J. Cosmol. Astropart. Phys.} \textbf{2025}, \emph{2025}, 062. \href{https://doi.org/10.1088/1475-7516/2025/04/062}{https://doi.org/10.1088/1475-7516/2025/04/062}.

\bibitem{Lutfuoglu:2025hjy}
Lütfüoğlu, B.C. Long-lived quasinormal modes and gray-body factors of black holes and wormholes in dark matter inspired Weyl gravity. \emph{Eur. Phys. J. C} \textbf{2025}, \emph{85}, 486. \href{https://doi.org/10.1140/epjc/s10052-025-14210-0}{https://doi.org/10.1140/epjc/s10052-025-14210-0}.

\bibitem{Kouniatalis:2025pxs}
Kouniatalis, G.; González, P.A.; Papantonopoulos, E.; et al. Near-horizon geodesic instabilities and anomalous decay of quasinormal modes in Weyl black holes. \emph{Phys. Rev. D} \textbf{2025}, \emph{112}, 124061. \href{https://doi.org/10.1103/ftvr-pclv}{https://doi.org/10.1103/ftvr-pclv}.
	\end{thebibliography}
\end{document}